\documentclass{article}
\usepackage{amsmath}
\usepackage{amssymb}
\usepackage{amsfonts}
\usepackage{amsthm}
\usepackage{graphicx}

\newcommand\pd[2]{{#1}_{,{#2}}}							
\newcommand\td[2]{\frac{d{#1}}{d{#2}}}					
\newcommand\op{\prime}									

\newcommand{\bvec}[1]{\boldsymbol{#1}}					
\newcommand{\uvec}[1]{\boldsymbol{\delta}_{#1}}			
\newcommand{\un}{\boldsymbol{\delta}_{n}}					
\newcommand{\tensor}[1]{\boldsymbol{{#1}}}				
\newcommand{\dyad}[1]{\tensor{#1}}						
\newcommand{\udyad}{\dyad{\delta}}						

\newcommand{\dpi}{\dyad{\pi}}

\newcommand{\dtau}{\dyad{\tau}}
\newcommand{\dros}{\tensor{\dot{\gamma}}}			

\newcommand{\pdt}[2]{\frac{\partial{#1}}{\partial{#2}}}		

\newcommand{\bnabla}{\boldsymbol{\nabla}}
\newcommand{\bcdot}{\boldsymbol{\cdot}}
\newcommand{\grad}[1]{\bvec{\bnabla}{#1}}
\newcommand{\dive}[1]{\bvec{\bnabla}\bcdot{#1}}

\newcommand{\Laplacian}[1]{\Delta{#1}}

\newcommand{\bq}[1]{{#1}^{(0)}}					
\newcommand{\pq}[1]{{#1}^{(1)}}					
\newcommand{\spr}[2]{\bvec{{#1}}\bcdot\bvec{{#2}}}	
\newcommand{\ccj}[1]{\bar{#1}}					
\newcommand{\intInfT}{\int_{-\infty}^t}			
\newcommand\Real{\mbox{Re}} 
\newcommand\Imag{\mbox{Im}} 

\begin{document}
\title{Rayleigh-Taylor Instability in Viscoelastic Fluids}
\author{Amey Joshi\\Tata Consultancy Services,\\International Technology Park,\\Whitefield, Bangalore, 560 066\\}
\maketitle

\begin{abstract}
In this paper I analyze the onset of Rayleigh-Taylor instability between two linear viscoelastic fluids assuming that the perturbations at the interface are small. In the first half, the
paper analyzes a stratified viscoelastic fluid in which I prove that the perturbations rise or fall exponentially without oscillating. The second half of the paper examines the effect of
electric and magnetic fields on viscoelastic fluids. I show that it is possible to choose electric or magnetic field gradient such that the effective acceleration due to gravity is zero.
If a heavy Newtonian fluid rests on top of a lighter Newtonian fluid such a choice of field gradient would have rendered the arrangement stable. If the fluids are viscoelastic, I show 
that a similar arrangement is unstable. 
\end{abstract}

\section{Introduction}\label{sec:intro}
An arrangement of a heavy fluid on top of a lighter, incompressible fluid is kinematically possible. However, it is known to be unstable. Lord Rayleigh \cite{jws} first analyzed it for
inviscid fluids in 1883. Taylor showed that the nature of instability of a heavy fluid on top of a lighter one is the same as a lighter fluid accelerating against a heavier one. 
Chandrasekhar, in his treatise on hydrodynamic stability \cite{sc}, studies the same instability in Newtonian fluids. Joseph, Beavers and Funanda \cite{jbf} report an experiment in 
which they observed that the growth rates of the most unstable modes in a viscoelastic drop are higher than in drops of Newtonian fluids. One therefore suspects that the 
elasticity of fluids is responsible for making the drops more susceptible to break up. Joseph, Beavers and Funanda, in the same paper, provide an approximate analysis using the viscous
potential flow technique to explain the phenomenon for a drop of an Oldroyd-B fluid. In this paper, I shall extend Chandrasekhar's analysis \cite{sc} to linear viscoelastic fluids. 
The same analysis is applicable to quasi-linear viscoelastic fluids under the viscous potential flow approximation. This is because the Jaumann derivatives in the constitutive 
relations of quasi-linear models are identical to convective derivatives of linear models.

\section{Viscoelastic fluids subject to small perturbations}\label{sec:sp}
Constitutive relations of linear viscoelastic fluids can be written in an integral form as \cite{bah}
\begin{equation}\label{sp:e1}
\dtau = -\int_{-\infty}^{t}G(t - t^\op)\dros(t^\op)dt^\op,
\end{equation}
where $G(t - t^\op)$ is the relaxation modulus, $\dtau$ is deviatoric part of the total stress tensor $\dpi$ and $\dros$ is the rate of strain tensor. I shall denote tensors by 
boldface Greek symbols and vectors as boldface Roman symbols. The only exception being unit vectors, which shall be denoted as a bold delta with a subscript. For example, $\uvec{x}$ is
the unit vector in the $x$ direction. I shall also use a convention, popular in Physics, of writing the total stress tensor as,
\begin{equation}\label{sp:e2}
\dpi = p\udyad + \dtau,
\end{equation}
where $p$ is the pressure and $\udyad$ is a unit tensor or rank 2. The law of conservation of momentum for viscoelastic fluids (actually, for all continua) is given by Cauchy equation,
\begin{equation}\label{sp:e3}
\rho\dot{\bvec{u}} = -\dive{\tensor{\pi}} + \bvec{f}
\end{equation}
In this equation, $\rho$ is the density of the fluid, $\bvec{u}$ is the velocity of a fluid parcel and $\bvec{f}$ is the density of a body force. $\dot{\bvec{u}}$ is the convective 
derivative of $\bvec{u}$. Using equations (\ref{sp:e1}) and (\ref{sp:e2}) in equation (\ref{sp:e3}) and assuming that the only body force is that due to gravity, 
\begin{equation}\label{sp:e4}
\rho\left(\pdt{\bvec{u}}{t} + \bvec{u}\bcdot\grad{\bvec{u}}\right) = -\grad{p} + \int_{-\infty}^{t}G(t - t^\op)\dive{\dros}(t^\op)dt^\op + \rho\bvec{g}
\end{equation}
Let the velocity field be perturbed so that it is expressed as $\bvec{u} = \bq{\bvec{u}} + \pq{\bvec{u}}$, the superscript '(0)' indicating the flow without perturbation and the 
superscript '(1)' indicating perturbation itself. Let the pressure and density be written in a similar fashion. If the perturbations are small, quadratic and higher order terms in 
perturbed quantities can be ignored to give,
\begin{eqnarray}\label{sp:e5}
\bq{\rho}\left(\pdt{\pq{\bvec{u}}}{t} + \bq{\bvec{u}}\bcdot\grad{\pq{\bvec{u}}} + \pq{\bvec{u}}\bcdot\grad{\bq{\bvec{u}}}\right) &  &	\nonumber		\\
+ \pq{\rho}\left(\pdt{\bq{\bvec{u}}}{t} + \bq{\bvec{u}}\bcdot\grad{\bq{\bvec{u}}}\right) &=& -\grad{\pq{p}} + \pq{\rho}\bvec{g} 	\\
 & & + \int_{-\infty}^{t}G(t - t^\op)\dive{\pq{\dros}}dt^\op \nonumber 
\end{eqnarray}
where $\pq{\dros}$ is the rate of strain tensor due to perturbation alone. Since the fluid is initially at rest, $\bq{\bvec{u}} = 0$ and equation (\ref{sp:e5}) gives,
\begin{equation}\label{sp:e6}
\bq{\rho}\pdt{\pq{\bvec{u}}}{t} = -\grad{\pq{p}} + \int_{-\infty}^{t}G(t - t^\op)\dive{\pq{\dros}}dt^\op + \pq{\rho}\bvec{g}
\end{equation}
If the fluid is incompressible, $\dive{\pq{\dros}} = \Laplacian{\pq{\bvec{u}}}$, $\Delta$ being the Laplacian operator. Therefore,
\begin{equation}\label{sp:e7}
\bq{\rho}\pdt{\pq{\bvec{u}}}{t} = -\grad{\pq{p}} + \int_{-\infty}^{t}G(t - t^\op)\Laplacian{\pq{\bvec{u}}}(t^\op)dt^\op + \pq{\rho}\bvec{g},
\end{equation}
Assume a normal mode expansion of the form $\pq{\bvec{u}} = \bvec{a}\exp(nt+i\spr{k}{x})$, where $n$ is complex and only the real part of the expression is physically significant. The
perturbation in pressure field is similarly written as $\pq{p} = b\exp(nt+i\spr{k}{x})$. Substituting it in the above equation we get
\begin{equation}\label{sp:e8}
n\bq{\rho}\pq{\bvec{u}} = -\bvec{k}\pq{p} - k^2\int_{-\infty}^{t}G(t - t^\op)\pq{\bvec{u}}(t^\op)dt^\op + \pq{\rho}\bvec{g}
\end{equation}
Equation (\ref{sp:e8}) can be written in a simpler form as
\begin{equation}\label{sp:e9}
n\bq{\rho}\pq{\bvec{u}} = -\bvec{k}\pq{p} - k^2\ccj{\eta}\pq{\bvec{u}} + \pq{\rho}\bvec{g}
\end{equation}
where the complex viscosity, $\ccj{\eta}$ \cite{bah} (overhead bar denotes the complex conjugate) is defined as
\begin{equation}\label{sp:e10}
\ccj{\eta} = \int_{0}^{\infty}G(s)\exp(-ns)ds
\end{equation}
If I had assumed the fluid to be Newtonian, I would have arrived at a similar equation except that instead of the of the complex viscosity $\ccj{\eta}$, I would have had $\mu$, the
viscosity at zero shear rate. Thus the entire machinery, developed for analyzing Rayleigh-Taylor instability in Newtonian fluids, can be reused for general viscoelastic fluids be 
replacing $\mu$ with $\ccj{\eta}$. 

\section{Linear analysis}\label{sec:la}
We can include the effect of surface tension in equation (\ref{sp:e7}) by adding a term $-(\gamma\dive{\pq{\un}})\uvec{z}\delta(z-z_0)$, where $\gamma$ is the interfacial tension, 
$\pq{\un}$ is the normal to the perturbed vector, $\delta(\cdot)$ is the Dirac delta function and the unperturbed interface is at $z=z_0$ , to get
\begin{eqnarray}
\bq{\rho}\pdt{\pq{\bvec{u}}}{t} &=& -\grad{\pq{p}} + \int_{-\infty}^{t}G(t - t^\op)\Delta\pq{\bvec{u}}(t^\op)dt^\op + \pq{\rho}{\bvec{g}} - \nonumber \\ 
							& & (\gamma\dive{\pq{\un}})\uvec{z}\delta(z-z_0)	\label{la:e1}
\end{eqnarray}
If the interface is perturbed by a quantity $\xi$ then 
\begin{equation}\label{la:e1a}
\pq{\un} = -\frac{\pd{\xi}{x}\uvec{x} + \pd{\xi}{y}\uvec{y} - \uvec{z}}{\sqrt{\pd{\xi}{x}^2 + \pd{\xi}{y}^2 + 1}}
\end{equation}
where the subscript notation for partial derivatives is used. That is, $\pd{\xi}{x}$ is the partial derivative of $\xi$ with respect to $x$. Let us assume that the density of the fluid 
varies only in the $z$-direction. Let the components of perturbed velocity field $\pq{\bvec{u}}$ be $(\pq{u},\pq{v},\pq{w})$. Therefore the equation of continuity, after taking into 
account the incompressibility of the fluid, is
\begin{equation}\label{la:e2}
\pq{\pdt{\rho}{t} }+ \pq{w}\bq{\pdt{\rho}{z}} = 0
\end{equation}
The incompressibility condition, on its own, reads
\begin{equation}\label{la:e3}
\pq{\pdt{u}{x}} + \pq{\pdt{v}{y}} + \pq{\pdt{w}{z}} = 0
\end{equation}
Assuming a normal mode expansion of the form $\exp(nt+i\spr{k}{x})$, where $n$ is complex and only the real part of the expression is physically significant and writing 
equation (\ref{la:e1}) in component form
\begin{eqnarray}
\bq{\rho}n\pq{u} &=& -ik_x\pq{p} + \intInfT G(t-t^\op)(D^2 - k^2)\pq{u}(t^\op)dt^\op			\label{la:e4}	\\
\bq{\rho}n\pq{v} &=& -ik_y\pq{p} + \intInfT G(t-t^\op)(D^2 - k^2)\pq{v}(t^\op)dt^\op			\label{la:e5}	\\
\bq{\rho}n\pq{w} &=& -D\pq{p} + \intInfT G(t-t^\op)(D^2 - k^2)\pq{w}(t^\op)dt^\op 	\nonumber		\\ 
			   & & -\pq{\rho}{g} - k^2\gamma\pq{z}\delta(z-z_0)	\label{la:e6}
\end{eqnarray}
Similarly, equation (\ref{la:e3}) becomes
\begin{equation}\label{la:e7}
ik_x\pq{u} + ik_y\pq{v} + D\pq{w} = 0,
\end{equation}
where $D\equiv d/dz$ and equation (\ref{la:e2}) becomes
\begin{equation}\label{la:e2a}
n\pq{\rho} + \pq{w}D\bq{\rho} = 0
\end{equation}
Multiplying equation (\ref{la:e4}) by $-ik_x$, equation (\ref{la:e5}) by $-ik_y$, adding them and using $(\ref{la:e7})$ in the result, we get
\begin{equation}\label{la:e8}
\bq{\rho}nD\pq{w} = -k^2\pq{p} + \intInfT G(t-t^\op)(D^2 - k^2)D\pq{w}dt^\op,
\end{equation}
where $k^2 = (k_x^2 + k_y^2)$. Writing in terms of complex viscosity\footnote{Proved in the appendix.},
\begin{equation}\label{la:e8a}
\bq{\rho}nD\pq{w} = -k^2\pq{p} + \ccj{\eta}(D^2 - k^2)D\pq{w}
\end{equation}
Using equation (\ref{la:e2a}) we can express $\pq{\rho}$ as $-\pq{w}D\bq{\rho}/n$ in (\ref{la:e6}). If we also write the result in terms of complex viscosity, we get
\begin{equation}\label{la:e9}
\bq{\rho}n\pq{w} = -D\pq{p} + \ccj{\eta}(D^2 - k^2)\pq{w} + \frac{{g}\pq{w}D\bq{\rho}}{n} - k^2\gamma\pq{z}\delta(z-z_0)
\end{equation}
Let the fluid be confined between two rigid boundaries $z=0$ and $z=d$. Since the tangential and normal velocities vanish at solid boundaries, we have $\pq{w}=0$ and $D\pq{w}=0$ at 
both, $z=0$ and $z=d$. Equations (\ref{la:e8a}) and (\ref{la:e9}) are coupled eigenvalue relations for the functions $\pq{w}$ and $\pq{p}$ with eigenvalue $n$. Let $n_i$ and $n_j$ be two 
eigenvalues with the corresponding eigenfunctions $(\pq{w}_i, \pq{p}_i)$ and $(\pq{w}_j, \pq{p}_j)$. For the eigenvalue $n_i$, equation (\ref{la:e9}) is,
\begin{equation}\label{la:e10}
D\pq{p}_i = -n_i\bq{\rho}\pq{w_i} + \ccj{\eta}(D^2 - k^2)\pq{w}_i + \frac{{g}\pq{w}_iD\bq{\rho}}{n_i} - k^2\gamma\pq{z}\delta(z-z_0)
\end{equation}
Since 
\begin{equation}\label{la:e10a}
\pq{w}_i = \td{\pq{z}}{t},
\end{equation}
the normal mode expansion of $\pq{z}$ gives $\pq{w}_i = \pq{z}n_i$. Therefore,
\begin{equation}\label{la:e10b}
D\pq{p}_i = -n_i\bq{\rho}\pq{w_i} + \ccj{\eta}(D^2 - k^2)\pq{w}_i + \frac{{g}\pq{w}_iD\bq{\rho}}{n_i} - k^2\gamma\frac{\pq{w}_i}{n_i}\delta(z-z_0)
\end{equation}
Multiply by $\pq{w}_j$ and integrate between $z=0$ and $z=d$, to get
\begin{eqnarray}\label{la:e11}
\int_0^d \pq{w}_j D\pq{p}_i dz	&=& -n_i\int_0^d\bq{\rho}\pq{w_i}\pq{w}_jdz + \ccj{\eta}\int_0^d[(D^2 - k^2)\pq{w}_i]\pq{w}_jdz + \nonumber		\\
							& & \frac{{g}}{n_i}\int_0^d \pq{w}_j\pq{w}_iD\bq{\rho}dz - \frac{k^2\gamma}{n_i}\int_0^d \pq{w}_j\pq{w}_i\delta(z-z_0)dz	\nonumber 	\\
							& &	
\end{eqnarray}
The boundary conditions on $\pq{w}_j$ allow us to write the left hand side as
\begin{equation}\nonumber
-\int_0^d \pq{p}_iD\pq{w}_j dz
\end{equation}
which, after substituting for $\pq{p}_i$ from equation (\ref{la:e8a}) becomes
\begin{equation}\nonumber
-\int_0^d \left\{\left[\frac{-n_i\bq{\rho}}{k^2} + \frac{\ccj{\eta}}{k^2}(D^2 - k^2)\right]D\pq{w}_i\right\} D\pq{w}_j dz = I_1 + I_2,
\end{equation}
where
\begin{eqnarray}
I_1 &=& \int_0^d \left[\frac{n_i\bq{\rho}}{k^2} + \ccj{\eta}\right]D\pq{w}_i D\pq{w}_j dz	\nonumber		\\
I_2 &=& -\int_0^d \left\{\frac{\ccj{\eta}}{k^2} D^2(D\pq{w}_i)\right\}D\pq{w}_jdz
\end{eqnarray}
The form of $I_2$ can be simplified after integrating by parts to
\begin{equation}\nonumber
I_2 = \frac{\ccj{\eta}}{k^2} \int_0^d (D^2\pq{w}_i)(D^2\pq{w}_j)dz
\end{equation}
Equation (\ref{la:e11}), therefore becomes
\begin{eqnarray}
\int_0^d \left[\frac{n_i\bq{\rho}}{k^2} + \ccj{\eta}\right]D\pq{w}_i D\pq{w}_j dz + \frac{\ccj{\eta}}{k^2} \int_0^d (D^2\pq{w}_i)(D^2\pq{w}_j)dz &=&	\nonumber		\\
-n_i\int_0^d\bq{\rho}\pq{w_i}\pq{w}_jdz + \ccj{\eta}\int_0^d[(D^2 - k^2)\pq{w}_i]\pq{w}_jdz  & &		\nonumber		\\
+ \frac{{g}}{n_i}\int_0^d \pq{w}_j\pq{w}_iD\bq{\rho}dz - \frac{k^2\gamma}{n_i}\int_0^d \pq{w}_j\pq{w}_i\delta(z-z_0)dz & &
\end{eqnarray}
Therefore,
\begin{eqnarray}
-n_i\int_0^d\left\{\bq{\rho}\pq{w_i}\pq{w}_j + \frac{\bq{\rho}}{k^2}D\pq{w}_i D\pq{w}_j\right\}dz + \frac{{g}}{n_i}\int_0^d \pq{w}_j\pq{w}_iD\bq{\rho}dz & & \nonumber 		\\
 - \frac{k^2\gamma}{n_i}\int_0^d \pq{w}_j\pq{w}_i\delta(z-z_0)dz &=&	\nonumber		\\
\ccj{\eta}\int_0^d\left\{D\pq{w}_i D\pq{w}_j + \frac{1}{k^2}(D^2\pq{w}_i)(D^2\pq{w}_j) - ((D^2 - k^2)\pq{w}_i)\pq{w}_j\right\}dz	\nonumber		\\
& &
\end{eqnarray}
Simplifying the last term on the right hand side and integrating the result by parts we have,
\begin{eqnarray}\label{la:e12}
-n_i\int_0^d\left\{\bq{\rho}\pq{w_i}\pq{w}_j + \frac{\bq{\rho}}{k^2}D\pq{w}_i D\pq{w}_j\right\}dz + \frac{{g}}{n_i}\int_0^d \pq{w}_j\pq{w}_iD\bq{\rho}dz & & \nonumber 		\\
 - \frac{k^2\gamma}{n_i}\int_0^d \pq{w}_j\pq{w}_i\delta(z-z_0)dz &=&	\nonumber		\\
\ccj{\eta}\int_0^d\left\{k^2\pq{w}_i\pq{w}_j + 2D\pq{w}_i D\pq{w}_j + \frac{1}{k^2}(D^2\pq{w}_i)(D^2\pq{w}_j)\right\}dz	\nonumber		\\
& &
\end{eqnarray}
One form of this equation is
\begin{eqnarray}\label{la:e12a}
{g}\int_0^d \pq{w}_j\pq{w}_iD\bq{\rho}dz - k^2\gamma\int_0^d \pq{w}_j\pq{w}_i\delta(z-z_0)dz &=& \nonumber	  \\
n_i^2\int_0^d\left\{\bq{\rho}\pq{w_i}\pq{w}_j + \frac{\bq{\rho}}{k^2}D\pq{w}_i D\pq{w}_j\right\}dz +	& & \nonumber		\\
n_i\ccj{\eta}\int_0^d\left\{k^2\pq{w}_i\pq{w}_j + 2D\pq{w}_i D\pq{w}_j + \frac{(D^2\pq{w}_i)(D^2\pq{w}_j)}{k^2}\right\}dz & &	\nonumber		\\
 & &
\end{eqnarray}
Interchanging $i$ and $j$,
\begin{eqnarray}\label{la:e13}
{g}\int_0^d \pq{w}_j\pq{w}_iD\bq{\rho}dz - k^2\gamma\int_0^d \pq{w}_j\pq{w}_i\delta(z-z_0)dz &=& \nonumber	  \\
 n_j^2\int_0^d\left\{\bq{\rho}\pq{w_i}\pq{w}_j + \frac{\bq{\rho}}{k^2}D\pq{w}_i D\pq{w}_j\right\}dz + & &	\nonumber		\\
 n_j\ccj{\eta}\int_0^d\left\{k^2\pq{w}_i\pq{w}_j + 2D\pq{w}_i D\pq{w}_j + \frac{(D^2\pq{w}_i)(D^2\pq{w}_j)}{k^2}\right\}dz & &	\nonumber		\\
 & &
\end{eqnarray}
If $n_j=\ccj{n}_i$, then
\begin{eqnarray}\label{la:e14}
{g}\int_0^d|\pq{w}_i|^2D\bq{\rho}dz - k^2\gamma\int_0^d|\pq{w}_i|^2\delta(z-z_0)dz &=& \nonumber	  \\
 (\ccj{n}_i)^2\int_0^d\left\{\bq{\rho}|\pq{w}_i|^2+\frac{\bq{\rho}}{k^2}|D\pq{w}_i|^2\right\}dz + & & \nonumber	 \\
 \ccj{n}_i\eta\int_0^d\left\{k^2|\pq{w}_i|^2 + 2|D\pq{w}_i|^2 + \frac{|D^2\pq{w}_i|^2}{k^2}\right\}dz & &	\nonumber		\\
 & &
\end{eqnarray}
Under the same assumptions, equation (\ref{la:e12a}) becomes
\begin{eqnarray}\label{la:e15}
{g}\int_0^d|\pq{w}_i|^2D\bq{\rho}dz - k^2\gamma\int_0^d|\pq{w}_i|^2\delta(z-z_0)dz &=& \nonumber	  \\
 (n_i)^2\int_0^d\left\{\bq{\rho}|\pq{w}_i|^2+\frac{\bq{\rho}}{k^2}|D\pq{w}_i|^2\right\}dz + & & \nonumber	 \\
 n_i\ccj{\eta}\int_0^d\left\{k^2|\pq{w}_i|^2 + 2|D\pq{w}_i|^2 + \frac{|D^2\pq{w}_i|^2}{k^2}\right\}dz & &	\nonumber		\\
 & &
\end{eqnarray}
Subtracting equation (\ref{la:e14}) from equation (\ref{la:e15}),
\begin{eqnarray}\label{la:e16}
(n_i^2 - (\ccj{n}_i)^2)\int_0^d\left\{\bq{\rho}|\pq{w}_i|^2+\frac{\bq{\rho}}{k^2}|D\pq{w}_i|^2\right\}dz &=&	\nonumber 	\\ 
(\ccj{n}_i\eta - n_i\ccj{\eta})\int_0^d\left\{k^2|\pq{w}_i|^2 + 2|D\pq{w}_i|^2 + \frac{|D^2\pq{w}_i|^2}{k^2}\right\}dz & &
\end{eqnarray}
The integrands in the above equation are positive definite and so are the limits, therefore the integrals themselves are positive. We can write equation (\ref{la:e16}) in a simpler
form as
\begin{equation}\label{la:e17}
(n_i^2 - (\ccj{n}_i)^2)I_3 = (\ccj{n}_i\eta - n_i\ccj{\eta})I_4,
\end{equation}
where $I_3 > 0$, $I_4 > 0$ and
\begin{eqnarray}
I_3 &=& \int_0^d\left\{\bq{\rho}|\pq{w}_i|^2+\frac{\bq{\rho}}{k^2}|D\pq{w}_i|^2\right\}dz	\label{la:e18} \\
I_4 &=& \int_0^d\left\{k^2|\pq{w}_i|^2 + 2|D\pq{w}_i|^2 + \frac{|D^2\pq{w}_i|^2}{k^2}\right\}dz	\label{la:e19}
\end{eqnarray}
If $n_i = a + bi$ then equation (\ref{la:e17}) immediately gives,
\begin{equation}\label{la:e20}
b = \frac{a\Imag(\eta)I_4}{2a I_3 + \Real(\eta)I_4}
\end{equation}
We obtained equation (\ref{la:e20}) by manipulating a form of equation (\ref{la:e12}). If we continue with that equation in its original form and interchange $i$ and $j$, we get
\begin{eqnarray}\label{la:e21}
-n_j\int_0^d\left\{\bq{\rho}\pq{w}_j\pq{w}_i + \frac{\bq{\rho}}{k^2}D\pq{w}_j D\pq{w}_i\right\}dz + \frac{{g}}{n_j}\int_0^d \pq{w}_i\pq{w}_jD\bq{\rho}dz & & \nonumber 		\\
 - \frac{k^2\gamma}{n_j}\int_0^d \pq{w}_i\pq{w}_j\delta(z-z_0)dz &=&	\nonumber		\\
\ccj{\eta}\int_0^d\left\{k^2\pq{w}_j\pq{w}_i + 2D\pq{w}_j D\pq{w}_i + \frac{1}{k^2}(D^2\pq{w}_j)(D^2\pq{w}_i)\right\}dz	\nonumber		\\
& &
\end{eqnarray}
Subtracting it from equation (\ref{la:e12}), we get
\begin{eqnarray}\label{la:e22}
(n_j - n_i)\int_0^d\left\{\bq{\rho}\pq{w}_j\pq{w}_i + \frac{\bq{\rho}}{k^2}D\pq{w}_j D\pq{w}_i\right\}dz + & & 	\nonumber 	\\
\left(\frac{{g}}{n_i} - \frac{{g}}{n_j}\right)\int_0^d \pq{w}_i\pq{w}_jD\bq{\rho}dz - & & \nonumber 	\\
\left(\frac{k^2\gamma}{n_i} - \frac{k^2\gamma}{n_j}\right)\int_0^d \pq{w}_i\pq{w}_j\delta(z-z_0)dz &=& 0
\end{eqnarray}
If we now choose $n_j = \ccj{n}_i$, we can write equation (\ref{la:e22}) as 
\begin{equation}\label{la:e23}
(\ccj{n}_i - n_i)I_5 + \left(\frac{{g}}{n_i} - \frac{{g}}{\ccj{n}_i}\right)I_6 - \left(\frac{k^2\gamma}{n_i} - \frac{k^2\gamma}{\ccj{n}_i}\right)I_7 = 0,
\end{equation}
where
\begin{eqnarray}
I_5 &=& \int_0^d\left\{\bq{\rho}|\pq{w}_i|^2 + \frac{\bq{\rho}}{k^2}|D\pq{w}_i|^2\right\}dz > 0	\label{la:e24}	\\
I_6 &=& \int_0^d |\pq{w}_i|^2 D\bq{\rho}dz \label{la:e25}	\\
I_7 &=& \int_0^d |\pq{w}_i|^2 \delta(z-z_0)dz > 0 \label{la:e26}	
\end{eqnarray}
The sign of $I_6$ depends on the density gradient. It is positive if density increases with increasing $z$ and is negative if density decreases with increasing $z$. We can simplify 
(\ref{la:e23}) to get
\begin{equation}\label{la:e27}
(\ccj{n}_i - n_i)\left(I_5 + \frac{{g}I_6}{|n_i|^2} - \frac{k^2\gamma I_7}{|n_i|^2}\right) = 0
\end{equation}

\subsection{Conclusions of the analysis}\label{sec:coa}
The main results of the analysis of section \ref{sec:la} are equations (\ref{la:e17}) and (\ref{la:e27}). Writing them together once again -
\begin{eqnarray}
(n_i^2 - (\ccj{n}_i)^2)I_3 &=& (\ccj{n}_i\eta - n_i\ccj{\eta})I_4 \label{la:e17r}	\\
(\ccj{n}_i - n_i)\left(I_5 + \frac{{g}I_6}{|n_i|^2} - \frac{k^2\gamma I_7}{|n_i|^2}\right) &=& 0 \label{la:e28r}		
\end{eqnarray}

\begin{enumerate}
\item If the fluids were Newtonian instead of viscoelastic, $\eta = \ccj{\eta}$ and equation (\ref{la:e17}) will imply
\begin{equation}\label{coa:e1}
\Real{(n_i)} = -I_4/I_3 < 0
\end{equation}
Equation (\ref{la:e17}) was derived under the assumption that $n_i$ was complex. Therefore, equation (\ref{coa:e1}) tells us that oscillatory modes, if they exist, are stable. It tallies 
with Chandrasekhar's \cite{sc} inference from his equation (86) of section (93). 

\item Suppose the fluids are viscoelastic and $I_6 > 0$, that is the density of fluid rises with $z$. Equation (\ref{la:e27}) implies that one of the following equations are true.
\begin{eqnarray}
(\ccj{n}_i - n_i) &=& 0 \label{la:e28}		\\
\left(I_5 + \frac{{g}I_6}{|n_i|^2} - \frac{k^2\gamma I_7}{|n_i|^2}\right) &=& 0 \label{la:e29}
\end{eqnarray}
If there were no surface tension, as in the case of a stratified fluid, equation (\ref{la:e29}) cannot be true. Therefore, the only way to satisfy equation (\ref{la:e27}) is to insist on 
truth of equation (\ref{la:e28}). Which means that if $I_6 > 0$, $n_i$ is real. 

\item Suppose the fluids are viscoelastic and $I_6 < 0$. Then equation (\ref{la:e29}) can be true and hence
\begin{equation}\label{coa:e2}
|n_i|^2 = (a^2 + b^2) = \left(\frac{k^2\gamma I_7 - {g}I_6}{I_5}\right)
\end{equation}
We now use the form of $b$ derived in equation (\ref{la:e20}) to get a quadric equation in $a$,
\begin{eqnarray}\label{coa:e3}
a^4I^2_3 + a^3I_3I_4\Real(\eta) + \left(|\eta|^2\frac{I^2_4}{4} - |n_i|^2I^2_3\right)a^2 - & & \nonumber 	\\
|n_i|^2I_3I_4\Real(\eta)a - \frac{|n_i|^2\Real^2(\eta)I^2_4}{4} &=& 0,
\end{eqnarray}
Since $|n_i|^2 > 0$, by Descartes' rule of signs, mentioned in appendix to this paper, the equation (\ref{coa:e3}) has at least one positive root. Therefore, oscillatory modes could be 
unstable. However, on physical grounds, we know that perturbations in a fluid stratified such that it is heavier at the bottom will always fade. Therefore, even in the case $I_6 < 0$, the
$n_i$ has to be real. Thus, there are no oscillatory modes.
\end{enumerate}

\section{Effect of electric and magnetic field}\label{eem}
I demonstrate the effect of gradient of an applied electric or magnetic field on the fluid's specific weight. To do so, I employ the example of linear analysis of instability of an 
interface between perfect fluids. Developing the framework for perfect fluids allows us to focus on the effect of fields without getting obscured by the peculiarities of non-
Newtonian constitutive relations. I shall also demonstrate that the mathematical structure of the problem is similar for the case of the electric and the magnetic field and that a single 
framework that analyzes both can be developed. 

Let us consider the stability of the interface between perfect fluids in presence of electric and magnetic fields. Let $\rho_1$ and $\rho_2$ be densities of lower and upper fluids 
respectively. Let both of them have a uniform base velocity $\bq{\bvec{u}}_1 = U_1\uvec{x}$ and $\bq{\bvec{u}}_2 = U_2\uvec{x}$ respectively, in the $X$ direction\footnote{Since I am 
dealing with ideal fluids, obeying Euler equations, any smooth function can represent a fully-developed steady flow.}. Let their interface be at $z = 0$. 

The mechanical energy of perfect fluids with irrotational flow, in presence of external conservative forces, can be described by the Bernoulli equation
\begin{equation}\label{eem:e17}
-\frac{\partial\phi}{\partial{t}} + \frac{u^2}{2} + \frac{p}{\rho} + \Phi = \mathcal{F}(t)
\end{equation}
where $\mathcal{F}(t)$ is an arbitrary function of time, $\Phi$ is potential due to external fields, ${\bvec{u}} = -\grad{\phi}$, $\phi$ is the velocity potential. The unperturbed velocity 
potential for lower fluid is $-U_1x$. Since the fluids are perfect and incompressible, Kelvin's vorticity theorem assures that if the flow was irrotational to begin with, the velocities 
induced because of perturbations will also be irrotational\cite{arh}. Let the perturbed potential be 
\begin{equation}\label{eem:e18}
\phi_1 = \bq{\phi}_1 + \pq{\phi}_1 = -Ux + \pq{\phi}_1
\end{equation}

Since the fluids are incompressible, $\Delta\phi_1 = 0$. Therefore equation (\ref{eem:e18}) implies that $\Delta\pq{\phi}_1 = 0$. Likewise, if $\pq{\phi}_2$ is the perturbed potential of 
the upper fluid, then $\Delta\pq{\phi}_2 = 0$. In terms of velocity potential of the lower fluid, the velocity of perturbed surface is 
\begin{equation}\label{eem:e19}
\pq{\dot{\xi}} = -\pq{\phi}_{1,z}
\end{equation}
where $\pq{\xi}$ is the $z$ coordinate of the perturbed interface, the overhead dot denotes the time derivative. Similar equation holds good for the upper fluid. An arbitrary perturbation 
of the surface and the corresponding driving velocity potentials can be expanded in terms of normal modes as
\begin{eqnarray}
\pq{\xi} &=& A\exp(nt + ikx)	\label{eem:e20}	\\
\pq{\phi}_{1} &=& C_1\exp(nt + ikx + kz)	\label{eem:e21}	\\
\pq{\phi}_{2} &=& C_2\exp(nt + ikx - kz)	\label{eem:e22}
\end{eqnarray}
The perturbations are written in a manner so that they vanish away from the interface. Substituting equations (\ref{eem:e20}) to (\ref{eem:e22}) in (\ref{eem:e19}) and its analog for 
upper fluid, at $z = 0$.
\begin{eqnarray}
-kC_1 &=& nA + ikU_1A	\label{eem:e23}	\\
kC_2 &=& nA + ikU_2A 	\label{eem:e24}
\end{eqnarray}
The third equation for finding the unknowns in equations (\ref{eem:e20}) to (\ref{eem:e22}) follows from the continuity of pressure across the interface. Assuming a surface tension 
$\gamma$ at the interface, we have at $z=0$, $p_1 = p_2 - \gamma\pq{\xi}_{,xx}$, where $p_1$ and $p_2$ are total pressure in lower and upper fluids respectively.  For an incompressible 
dielectric fluid in electric and gravitational fields, 
\begin{equation}\label{eem:e25}
\Phi_1 = \Phi_{1e} = g\xi + \frac{\epsilon_0}{2\rho_1}E_1^{2}\left(\frac{\partial{\kappa_1}}{\partial{\rho_1}}\right)_{\theta}
\end{equation}
and for a ferro-fluid in magnetic and gravitational fields,
\begin{equation}\label{eem:e26}
\Phi_1 = \Phi_{1m} = g\xi + \frac{\mu_0}{2\rho_1}H_1^{2}\left(\frac{\partial{\kappa_{1m}}}{\partial{\rho_1}}\right)_{\theta}
\end{equation}
where $\theta$ is the temperature. The second terms on the right hand sides of equations (\ref{eem:e25}) and (\ref{eem:e26}) arise due to electro and magneto striction forces \cite{pp}. 
Since the form of equations (\ref{eem:e25}) and (\ref{eem:e26}) is similar, I will use one of them, the magnetic one, in the analysis and merely quote the results for the electric case. 
Using equation (\ref{eem:e26}) in equation (\ref{eem:e17}) and using the boundary condition on pressure,  
\begin{eqnarray}\label{eem:e27}
\rho_1\left(-\pdt{\phi_1}{t}+\frac{v_1^2}{2}+g\pq{\xi}+K_{1m}^2 H_1^2\right) &=& \rho_2\left(-\pdt{\phi_2}{t}+\frac{v_2^2}{2}+g\pq{\xi}+K_{2m}^2 H_2^2\right)+ \nonumber 	\\
 & & \mathcal{F}(t) - \gamma\frac{\partial^2 \pq{\xi}}{\partial x^2}
\end{eqnarray}
where $\mathcal{F}(t) = \rho_1{F_1(t)} - \rho_2{F_2(t)}$ and the constants $K_{1m}$ and $K_{2m}$ are given by
\begin{eqnarray}
K_{1m}^2 &=& \frac{\mu_0}{2\rho_1}\left(\frac{\partial\kappa_{1m}}{\partial\rho_1}\right)_{\theta}	\\	\label{eem:e27a}
K_{2m}^2 &=& \frac{\mu_0}{2\rho_2}\left(\frac{\partial\kappa_{2m}}{\partial\rho_2}\right)_{\theta}		\label{eem:e27b}
\end{eqnarray}
Under unperturbed conditions, with $v_1 = U_1$, $v_2 = U_2$, $\pq{\phi}_1 = 0$, $\pq{\phi}_2 = 0$ and $\pq{\xi} = 0$, equation (\ref{eem:e27}) becomes
\begin{equation}\label{eem:e28}
\rho_1\left[\frac{U_1^{2}}{2} + {K_{1m}^2}H^2(0^{-})\right] = \rho_2\left[\frac{U_2^{2}}{2} + {K_{2m}^2}H^2(0^{+})\right] + \mathcal{F}(t) 
\end{equation}
where $\bvec{H}(0^{-})$ and $\bvec{H}(0^{+})$ are magnetizing fields just below and above the interface. Substituting $\mathcal{F}(t)$ from the above equation in (\ref{eem:e27}), 
noting that $v_1=U_1+\pq{v}_1$, $v_2=U_2+\pq{v}_2$ and using the linear (small amplitude) approximation,
\begin{equation}\label{eem:e29}
\rho_1\left[-\pq{\phi}_{1,t} - U_1\pq{\phi}_{1,x} + \breve{g}_{1m}\pq{\xi}\right] = \rho_2\left[-\pq{\phi}_{2,t} - U_2\pq{\phi}_{2,x} + \breve{g}_{2m}\pq{\xi}\right]
\end{equation}
where the "effective acceleration due to gravity" in each fluid is given by
\begin{equation}\label{eem:e30}
\breve{g}_{1m} = g - {K_{1m}^2}H(0^{-})H_{1,z}
\end{equation}
and
\begin{equation}\label{eem:e31}
\breve{g}_{2m} = g - {K_{2m}^2}H(0^{+})H_{2,z}
\end{equation}
Using equations (\ref{eem:e20}) to (\ref{eem:e22}) in (\ref{eem:e29}) we get the third equation in the unknowns $A, C$ and $C^{\op}$ from which we get the dispersion relation
\begin{equation}\label{eem:e32}
\frac{in}{k} = \frac{\rho_1 U_1 + \rho_2 U_2}{\rho_1 + \rho_2} \pm \left[\frac{\rho_1\breve{g}_{1m} - \rho_2\breve{g}_{2m}}{k(\rho_1 + \rho_2)} + \frac{k\gamma}{\rho_1 + 
\rho_2} - \frac{\rho_1\rho_2(U_1 - U_2)^2}{(\rho_1 + \rho_2)^2}\right]^{1/2}
\end{equation}

In absence of magnetic field $\breve{g}_{1m} = \breve{g}_{2m} = g$ and the dispersion relation reduces to the one on Lamb's treatise \cite{hl}. Thus the effect of an applied magnetic field 
gradient is to alter the specific weight $\rho\breve{g}$. By choosing the value of the gradient $\pd{H}{z}$, we can effectively increase or decrease specific weight to our advantage.

We get the relations for electric field gradients from equations (\ref{eem:e30}) and (\ref{eem:e31}) by replacing $\mu_0$ by $\epsilon_0$, $\kappa_m$ by $\kappa$, $K_m$ by $K_e$ and $H$ 
by $E$. They are
\begin{equation}\label{eem:e30a}
\breve{g}_{1e} = g - {K_{1e}^2}E(0^{-})E_{1,z}
\end{equation}
and
\begin{equation}\label{eem:e31a}
\breve{g}_{2e} = g - {K_{2e}^2}E(0^{+})E_{2,z}
\end{equation}
The constants $K_{1e}$ and $K_{2e}$, given by equations (\ref{eem:e30b}) and (\ref{eem:e31b}), depend only on the molecular and bulk properties of the fluids. 
\begin{eqnarray}
K_{1e}^2 &=& \frac{\epsilon_0}{2\rho_1}\left(\frac{\partial\kappa_{1e}}{\partial\rho_1}\right)_{\theta}	\\	\label{eem:e30b}
K_{2e}^2 &=& \frac{\epsilon_0}{2\rho_2}\left(\frac{\partial\kappa_{2e}}{\partial\rho_2}\right)_{\theta}		\label{eem:e31b}
\end{eqnarray}
The dispersion relation, likewise, will be 
\begin{equation}\label{eem:e33}
\frac{in}{k} = \frac{\rho_1 U_1 + \rho_2 U_2}{\rho_1 + \rho_2} \pm \left[\frac{\rho_1\breve{g}_{1e} - \rho_2\breve{g}_{2e}}{k(\rho_1 + \rho_2)} + \frac{k\gamma}{\rho_1 + 
\rho_2} - \frac{\rho_1\rho_2(U_1 - U_2)^2}{(\rho_1 + \rho_2)^2}\right]^{1/2}
\end{equation}

The form of equations (\ref{eem:e30}), (\ref{eem:e31}), (\ref{eem:e30a}) and (\ref{eem:e31a}) tells that it is possible to choose magnetic and electric fields gradients such that the
effective acceleration due to gravity is zero. Doing so will have an interesting impact on the analysis of section \ref{sec:la}.

\subsection{Analysis of section \ref{sec:la} in presence of electric and magnetic fields}\label{sec:aos}
The main results of the analysis of section \ref{sec:la} are equations (\ref{la:e17}) and (\ref{la:e27}). Writing them together once again -
\begin{eqnarray}
(n_i^2 - (\ccj{n}_i)^2)I_3 &=& (\ccj{n}_i\eta - n_i\ccj{\eta})I_4 \label{la:e17r1}	\\
(\ccj{n}_i - n_i)\left(I_5 + \frac{\breve{g}I_6}{|n_i|^2} - \frac{k^2\gamma I_7}{|n_i|^2}\right) &=& 0 \label{la:e28r2}
\end{eqnarray}
where $\breve{g}$ is the effective acceleration due to gravity in presence of electric ($\breve{g}_e$) or magnetic field ($\breve{g}_m$). It is given by one of the equations 
(\ref{eem:e30}), (\ref{eem:e31}), (\ref{eem:e30a}) and (\ref{eem:e31a}).

Suppose the fluids are viscoelastic and $I_6 > 0$ , that is the density of fluid rises with $z$. At the interface between immiscible fluids, there will always be a tension. Therefore, 
$\gamma$ cannot be zero. On the other hand, we can choose field gradient such that $\breve{g} = 0$ at the interface. In that case, equation (\ref{la:e28r2}) can be satisfied without 
forcing $n_i$ to be real. The real part of $n_i$ then satisfies the quadric equation (\ref{coa:e3}) that has at least one positive root. The corresponding perturbation modes are unstable.
If the fluids were Newtonian, $\breve{g} = 0$ would have guaranteed stability \cite{asj}.

\section{Appendix}
\theoremstyle{plain}
\newtheorem{thm}{Theorem}
\begin{thm}
If the eigenfunction $\pq{w}_i$ belongs to the eigenvalue $n_i$ of equations (\ref{la:e8a}) and (\ref{la:e10}) then ${\pq{\ccj{w}}_i}$ is an eigenfunction of
\begin{eqnarray}
-\bq{\rho}nD\pq{w} &=& k^2\pq{p} - \eta(D^2 - k^2)\pq{w}	\nonumber		\\
D\pq{p}_i &=& -n_i\bq{\rho}\pq{w_i} + \eta(D^2 - k^2)\pq{w}_i + \frac{{g}\pq{w}_iD\bq{\rho}}{n_i}	\nonumber
\end{eqnarray}
belonging to the eigenvalue $\ccj{n}_i$.
\end{thm}
\begin{proof}
By hypothesis
\begin{eqnarray}
-\bq{\rho}n_i D\pq{w}_i &=& k^2\pq{p} - \ccj{\eta}(D^2 - k^2)\pq{w}_i 	\nonumber		\\
D\pq{p}_i &=& -n_i\bq{\rho}\pq{w_i} + \ccj{\eta}(D^2 - k^2)\pq{w}_i + \frac{{g}\pq{w}_iD\bq{\rho}}{n_i}		\nonumber		\\
\end{eqnarray}
Taking complex conjugates of these equations and remembering that $\pq{\rho}$, $\pq{p}$, $\bvec{k}$ are all real functions, we have
\begin{eqnarray}
-\bq{\rho}\ccj{n}_i D{\pq{\ccj{w}}_i} &=& k^2\pq{p} - \eta(D^2 - k^2){\pq{\ccj{w}}_i} 	\nonumber		\\
D\pq{p}_i &=& -\ccj{n}_i\bq{\rho}{\pq{\ccj{w}_i}} + \eta(D^2 - k^2){\pq{\ccj{w}}_i} + \frac{{g}{\pq{\ccj{w}}_i}D\bq{\rho}}{\ccj{n}_i}	\nonumber		\\
& & 
\end{eqnarray}
\end{proof}

\begin{thm}
\begin{equation}\nonumber
\intInfT G(t-t^\op)(D^2 - k^2)\pq{w}dt^\op = \ccj{\eta}(D^2 - k^2)\pq{w},
\end{equation}
where symbols have their usual meanings.
\end{thm}
\begin{proof}
The normal mode expansion of $\pq{w}$ is of the form,
\begin{equation}\nonumber
\pq{w} = \pq{\breve{w}}(z)\exp(nt + i\spr{k}{x}).
\end{equation}
Therefore the left hand side (LHS) becomes,
\begin{eqnarray}
LHS &=& \intInfT G(t-t^\op)(D^2 - k^2)\pq{\breve{w}}(z)\exp(nt^\op + i\spr{k}{x})dt^\op	\nonumber 	\\
	&=& (D^2 - k^2)\pq{\breve{w}}(z)\exp(i\spr{k}{x})\intInfT G(t-t^\op)\exp(nt^\op)dt^\op	\nonumber 	
\end{eqnarray}
A transformation $s = t - t^\op$ simplifies the integral to
\begin{equation}\nonumber
\intInfT G(t-t^\op)\exp(nt^\op)dt^\op = \exp(nt)\int_{0}^{\infty}G(s)\exp(-ns)ds = \exp(nt)\ccj{\eta}
\end{equation}
Therefore,
\begin{equation}
LHS = (D^2 - k^2)\pq{\breve{w}}(z)\exp(i\spr{k}{x})\exp(nt)\ccj{\eta} = \ccj{\eta}(D^2 - k^2)\pq{w}
\end{equation}
\end{proof}

Descartes rule of signs, proved in ~\cite{xw}, is:
\begin{enumerate}
\item The rule states that if the terms of a single-variable polynomial with real coefficients are ordered by descending variable exponent, then the number of positive roots of the 
polynomial is either equal to the number of sign differences between consecutive nonzero coefficients, or is less than it by an even number. Multiple roots of the same value are counted 
separately.
\item As a corollary of the point 1, the number of negative roots is the number of sign changes after multiplying the coefficients of odd-power terms by −1, or fewer than it by a multiple 
of 2. 
\end{enumerate}
\bibliographystyle{amsplain}

\end{document}